\documentclass{article}
\usepackage{a4wide}
\pdfoutput=1
\usepackage{graphicx}
\usepackage{epsfig}
\usepackage{color}
\usepackage{transparent}
\usepackage{amsmath}
\usepackage{amssymb}
\usepackage{enumerate}
\usepackage{float}
\usepackage{color}
\usepackage{transparent}
\usepackage{enumerate}

\newcommand{\lungh}{w}

\newcommand{\commento}[1] {}
\newcommand{\comment}[1] {}

\newtheorem{theorem}{Theorem}
\newtheorem{property}{Property}
\newtheorem{lemma}[theorem]{Lemma}

\newtheorem{remark}{Remark}

\newtheorem{definition}{Definition}
\newcommand{\qed}{\hfill \ensuremath{\Box}}

\newenvironment{proof}{\vspace{1ex}\noindent{\bf Proof.}\hspace{0.5em}}
	{\hfill\qed\vspace{2ex}}

\newcounter{progcount}
\newcounter{linecount}[progcount]

{
    \end{tabbing}
    \end{minipage}\\[0.5ex]
}

\ignorespaces

\begin{document}

\title{Max-flow vitality in undirected unweighted planar graphs}

\author{Giorgio Ausiello\footnote{Dipartimento di Ingegneria Informatica, Automatica e Gestionale, Universit\`a
di Roma ``La Sapienza'', via Ariosto 25, 00185 Roma, Italy. Email: \texttt{ausiello@diag.uniroma1.it}.}
\and 
Paolo G. Franciosa\footnote{Dipartimento di Scienze Statistiche, Universit\`a di Roma ``La Sapienza'',
piazzale Aldo Moro 5, 00185 Roma, Italy. Email: \texttt{paolo.franciosa@uniroma1.it, isabella.lari@uniroma1.it}.} 
\and 
Isabella Lari\footnotemark[2]
\and 
Andrea Ribichini\footnote{Dipartimento Istituto Italiano di Studi Orientali - ISO, Universit\`a
di Roma ``La Sapienza'', Circonvallazione Tiburtina 4, 00185 Roma, Italy E-mail: \texttt{ribichini@diag.uniroma1.it}.}
}

\date{}

\maketitle

\begin{abstract}
We show a fast algorithm for determining the set of edges in a planar undirected unweighted graph, whose deletion reduces the maximum flow between two fixed vertices.
This is a special case of the  \emph{max flow vitality} problem, that has been efficiently solved for general undirected graphs and $st$-planar graphs.
The \emph{vitality} of an edge of a graph with respect to the maximum flow between two fixed vertices $s$ and $t$ is defined as the reduction of the maximum flow caused by the removal of that edge.
In this paper we show that the set of edges having vitality greater than zero in a planar undirected unweighted graph with $n$ vertices can be found in $O(n \log n)$ worst-case time and $O(n)$ space.
\end{abstract}
\vspace{1cm}
\textbf{Keywords:} maximum flow, minimum cut, vitality, planar graphs, undirected graphs, fault resiliency.

\section{Introduction}\label{se:intro}
Given a graph $G$, where each edge $e$ is allowed to carry a maximum amount $c(e)$ of \emph{flow}, and given two special vertices $s,t \in G$, the maximum flow (max-flow, for short) problem consists in determining a flow assignment to each edge so that the total flow from $s$ to $t$ is maximized. The flow assignment is subject to two constraints: the capacity bound for each edge, given by $c(e)$, and the conservation constraint for each vertex $v$ other than $s$ and $t$, stating that the total flow entering into $v$ must equal the total flow exiting from $v$.

A very wide literature has been produced since the 1950's about fast algorithms for computing the maximum flow value and/or a max-flow assignment for general graphs and restricted graph classes. We refer the reader to~\cite{Ahuja} for classical results and to~\cite{KingRaoTarjan,Orlin} for recent efficient max-flow algorithms for general graphs.

In the case of  undirected planar graphs,
Reif~\cite{Reif} proposed a divide and conquer approach for computing a minimum $st$-cut (hence a maximum flow by the well known duality theorem between min-cut and max-flow) in $O(n\log^2n)$ worst-case time, where $n$ is the number of vertices in the graph. By plugging in the linear time single-source shortest path (SSSP) tree algorithm for planar graphs by Henzinger et al.~\cite{Henzinger}, this bound can be improved to $O(n\log n)$. The best currently known approach for computing a minimum $st$-cut and a max-flow assignment is due to Italiano et al.~\cite{Italiano}, and it achieves $O(n\log \log n)$ time by a two phase approach, that exploits the algorithm by Hassin and Johnson~\cite{Hassin85}, that proposed an $O(n\log^2 n)$ worst-case time for computing a maximum flow assignment by solving a planar SSSP tree problem. Also in this case, the worst-case time bound can be improved to  $O(n \log n)$ by applying the algorithm in~\cite{Henzinger}.
If $G$ is directed and planar, the divide and conquer approach by Reif for the undirected case cannot be applied: for this case, Borradaile and Klein~\cite{Borradaile} presented an $O(n\log n)$ time algorithm, based on a repeated search of left-most circulations. In the case of directed planar unweighted graphs, Eisenstat and Klein presented a linear time algorithm~\cite{KleinUnweightedStoc2013}.

For directed $st$-planar graphs, i.e., graphs admitting a planar embedding with $s$ and $t$ on a same face, Hassin~\cite{Hassin} proved that both the max-flow value and a max-flow assignment can be found by computing an SSSP tree in the dual graph. By applying the algorithm by Henzinger \emph{et al.}~\cite{Henzinger} these problems can be solved in linear time.

\medskip
The \emph{vitality} of an edge $e$ with respect to max-flow measures the max-flow decrement observed when edge $e$ is removed from the graph.
A survey on vitality with respect to max-flow problems can be found in~\cite{AusielloNetworks}.
In the same paper, it is shown that:
\begin{itemize}
\item the vitality of all edges in a general undirected graph can be computed by solving $O(n)$ max-flow instances, thus giving an overall $O(n^{2}m)$ algorithm by applying the $O(mn)$ max-flow algorithms described in~\cite{KingRaoTarjan,Orlin};
\item for $st$-planar graphs (both directed or undirected) the vitality of all edges can be found in optimal $O(n)$ worst-case time. The same result holds for determining the vitality of all vertices, i.e., the max-flow reduction under the removal of all the edges incident on a same vertex;
\item the problem of determining the max-flow vitality of an edge is at least as hard as computing the max-flow for the graph, both for general graphs and for the restricted class of $st$-planar graphs.
\end{itemize}
The vitality problem is left open for general planar graphs, both directed and undirected.

\medskip
In this paper, we address the vitality problem in the case of unweighted undirected planar graphs. We propose a recursive algorithm that computes the vitality of all edges in $O(n\log n)$ worst-case time and $O(n)$ space.

This paper is organised as follows: in Section~\ref{se:defs} we provide some definitions and preliminary considerations, our results are presented in Section~\ref{se:planararcs}, while Section~\ref{se:conclusions} presents some final considerations and open problems.

\section{Definitions and preliminaries}
\label{se:defs}
Given a connected undirected  graph $G=(V,E)$ with $n$ vertices, we denote an edge $e=\{i,j\} \in E$ by the shorthand notation $ij$, and we define $\mbox{dist}(u, v)$ as the minimum number of edges in a path joining vertices $u$ and $v$.

We assume that each edge $e=ij \in E$ has an associated positive \emph{capacity} $c(e)$. 
Let $s \in V$ and $t \in V$, $s\neq t$, be two fixed vertices.
A \emph{feasible flow} in $G$ assigns to each edge $e=ij \in E$ two real values $x_{ij} \in [0,c(e)]$ and $x_{ji} \in [0,c(e)]$ such that:
$$\sum_{j:ij \in E} x_{ij} = \sum_{j:ij \in E} x_{ji}, \mbox{ for each } i \in V \setminus \{s,t\}.$$
The \emph{flow from $s$ to $t$} under a feasible flow assignment $x$ is defined as
$$F(x) = \sum_{j:sj \in E} x_{sj} - \sum_{j:sj \in E} x_{js}.$$
The $\emph{maximum flow}$ from $s$ to $t$ is the maximum value of $F(x)$ over all feasible flow assignments $x$.

An \emph{$st$-cut} is a partition of $V$ into two subsets $S$ and $T$ such that $s \in S$ and $t \in T$. The \emph{capacity of an $st$-cut} is the sum of the capacities of the edges $ij \in E$ 
such that $|S \cap \{i,j\}|=1$ and $|T \cap \{i,j\}|=1$.
The well known Min-Cut Max-Flow theorem \cite{Fulkerson} states that the maximum flow from $s$ to $t$ is equal to the capacity of a minimum $st$-cut for any weighted graph $G$.

\begin{definition}
The \emph{vitality} $\mbox{vit}(e)$ of an edge $e$ with respect to the maximum flow from $s$ to $t$ (or with respect to the minimum $st$-cut), according to the general concept of vitality in \cite{dagstuhl}, is defined as the maximum flow in $G$ minus the maximum flow in $G'=(V,E \setminus \{e\})$.
\end{definition}

The \emph{dual} of a planar embedded undirected graph $G$ is an undirected multigraph $G^*$, whose vertices correspond to faces of $G$ and such that for each edge $e$ in $G$ there is an edge $e^{*} = \{f^{*},g^{*}\}$   in $G^{*}$, where $f^{*}$ and $g^{*}$ are the vertices corresponding to the two faces $f$ and $g$ adjacent to $e$ in $G$. The length $\lungh(e^{*})$ of $e^{*}$ equals the capacity of $e$. 

We are given a path $\pi$ from $a$ to $b$ in a planar embedded graph and a vertex $v$ of $\pi$ different from $a$ and $b$. Fixing an orientation of $\pi$, for example from $a$ to $b$, we can say that every edge $\{v,w\}$ such that $w$ does not belong to $\pi$ \emph{lies to the left } or \emph{lies to the right} of $\pi$.

Given two paths $\pi_{1}, \pi_{2}$ in a planar embedded graph, a \emph{crossing} between $\pi_{1}$ and $\pi_{2}$ is a minimal subpath of $\pi_{1}$ defined by vertices $v_{1}, v_{2}, \ldots, v_{k}$, with $k \geq 3$, such that vertices $v_{2}, \ldots, v_{k-1}$ are contained in $\pi_{2}$, and, fixing an orientation of $\pi_{2}$, edge $\{v_{1}, v_{2}\}$ lies to the  left of 
$\pi_{2}$ and edge $\{v_{k-1}, v_{k}\}$ lies to the right of $\pi_{2}$. We say $\pi_{1}$ \emph{crosses
$\pi_{2}$ $t$ times} if there are $t$ different crossings between  $\pi_{1}$ and $\pi_{2}$. In a similar way, we can define a crossing between a cycle and a path.


\section{Vitality of edges in undirected unweighted planar graphs}
\label{se:planararcs}

We first observe that in graphs with integer edge capacities, the maximum flow value always is an integer number~\cite{Ahuja}, and edge vitality is defined as the difference between two maximum flows; since the vitality of an edge cannot exceed its capacity, in the unweighted case each edge may have vitality at most 1. Therefore we have:

\begin{remark} In an unweighted graph, both directed and undirected, edges can only have vitality 0 or 1.
\end{remark}

\noindent In particular, an edge has vitality 1 if and only if it appears in at least one minimum $st$-cut (i.e., it is essential to achieve maximum flow, for all flow assignments).

\medskip
We propose an $O(n \log n)$ time algorithm for computing the vitaliy of all edges in an unweighted undirected planar graph. Our algorithm is based on Reif's divide and conquer technique.

We assume a planar embedding of the graph is fixed, and we work on the dual graph $G^{*}$ defined by this embedding.
We fix in $G^{*}$ a face $f_{s}^{*}$ adjacent to $s$  and a face $f_{t}^{*}$  adjacent to $t$.
A cycle in the dual graph $G^{*}$ that separates face $f_{s}^{*}$ from face $f_{t}^{*}$ is called an \emph{$st$-separating cycle}.  Moreover, we choose a shortest path $\pi$ in $G^{*}$ from $f_{s}^{*}$ to $f_{t}^{*}$.

\begin{property}[see~\cite{Itai} and Propositions 1 and 2 in~\cite{Reif}]\label{pro:cycle}
An $st$-cut (resp., minimum $st$-cut) in $G$ corresponds to a cycle (resp., shortest cycle) in $G^{*}$ that separates face $f_{s}^{*}$ from face $f_{t}^{*}$.
\end{property}

\begin{property}[see Lemma 4.2 in~\cite{Itai} and Proposition 3 in~\cite{Reif}]\label{pro:crossundir}
Let $G$ be an undirected planar graph, and let $\pi$ be a shortest path from $f_s^{*}$ to $f_t^{*}$ in $G^{*}$: a shortest $st$-separating cycle $\gamma$ exists that crosses $\pi$ exactly once.
\end{property}
Thanks to Property~\ref{pro:crossundir}, the length of a shortest $st$-separating cycle $\gamma$  can be found as
\begin{equation*}
c(\gamma) = \min_{f^{*} \in \pi}\{c(\gamma_{f})\}
\end{equation*}
where $\gamma_{f}$ is a shortest $st$-separating cycle that contains vertex $f^{*}$.

Bearing in mind Property~\ref{pro:cycle}, edges with vitality 1 can be found by determining which edges, in the dual graph, belong to at least one minimum $st$-separating cycle. 

All minimum $st$-separating cycles must cross $\pi$ at least once, 
but minimum $st$-separating cycles might exist that cross $\pi$ many times.
Let $C$ be a multiple-crossing minimum $st$-separating cycle, and let $L $ be the set of vertices that belong to both $C$ and $\pi$.
Consider two vertices $a,b \in L$ such that:
\begin{itemize}
\item  path $C(a,b)$ in $C$ joining $a$ and $b$ does not contain any other vertex in $L$;
\item the union of $C(a,b)$ and the path $\pi(a,b)$ from $a$ to $b$ in $\pi$ is an $st$-separating cycle.
\end{itemize}
It can be seen that such vertices always exist, possibly with $a = b$, otherwise $C$ cannot be an $st$-separating cycle. Figure~\ref{fi:crossing} shows an example of an $st$-separating cycle      
that crosses $\pi$ more than once.

Let $C'(a,b) = C \setminus C(a,b)$:  note that all vertices in $L$ belong to $C'(a,b)$. The following properties hold. 
\begin{property}\label{prop:shortest}
Since $C$ is a minimum $st$-separating cycle, $C'(a,b)$ is a shortest path from $a$ to $b$.
\end{property}
\begin{property}\label{prop:seq}
All vertices in $L$  appear walking on $C'(a,b)$ from $a$ to $b$  in the same order, or the reversal, as they appear walking on $\pi$.
\end{property}
Property~\ref{prop:seq} directly derives from the fact that both $\pi$ and $C'(a,b)$ are shortest paths.
Thus, any minimum $st$-separating cycle crosses $\pi$ as depicted in Figure~\ref{fi:crossing}, where the crossing order along $\pi$ is the same as along $C'(a,b)$, while configurations shown in Figure~\ref{fi:crossingsforbidden} are forbidden for minimum $st$-separating cycles. Moreover, for each pair $x,y$ of vertices in $L$, the portion of $C$ from $x$ to $y$ (avoiding $C(a,b)$) and the portion of $\pi$ from $x$ to $y$ have the same length.

\begin{figure}[t]
\begin{center}
\def\svgwidth{12cm}
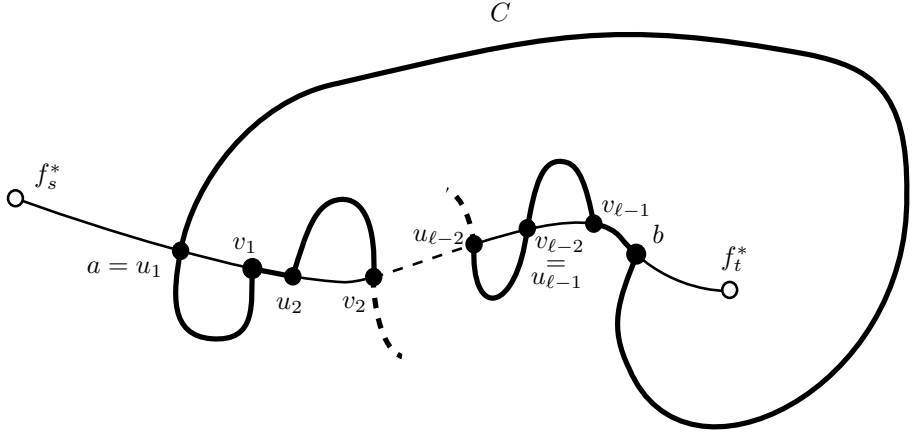
\end{center}
\caption{A multiple-crossing $st$-separating cycle $C$.}\protect\label{fi:crossing}
\end{figure}	

The following lemma shows that, in order to find all edges that belong to at least one minimum $st$-separating cycle, we can restrict our attention to minimum $st$-separating cycles that cross $\pi$ exactly once. Let us call this cycles \emph{single-crossing} minimum $st$-separating cycles.
 
\begin{lemma}\label{le:crossings}
Each edge contained in a minimum $st$-separating cycle is also contained in some single-crossing minimum $st$-separating cycle.
\end{lemma}
\begin{proof}
Let $C$ be a minimum $st$-separating cycle and let $a$, $b$ and $C'(a,b)$ as defined above. Let $\ell$ be the number of crossings between $C$ and $\pi$.
We observe that being $C$ an $st$-separating cycle then $\ell$ must be odd. 

Suppose that  $C$ crosses $\pi$ more than once, hence $\ell \geq 3$.
Since both $\pi$ and $C'(a,b)$ are shortest paths, as a consequence of Property~\ref{prop:seq}, the crossings appear in the same order both in $\pi$ and in $C'(a,b)$.
Let us assign the indices $1, \ldots, \ell$ to the crossings following the order in which they appear in both paths.
W.l.o.g., we refer to the example in Figure~\ref{fi:crossing}, where vertex $f^{*}_{t}$ is in the finite region defined by $C$. 

For each crossing $i$, $i=1, \ldots, \ell -1$, there is a subpath $C'(u_{i},v_{i})$ of $C'(a,b)$ from $u_{i}$ to $v_{i}$ such that $u_{i}$ and $v_{i}$ are the only vertices belonging also to $\pi$ (see Figure~\ref{fi:crossing}). Let us denote by $\pi(u_{i},v_{i})$ the corresponding subpath of $\pi$. Note that $u_{1}$ might be vertex $a$, $v_{\ell-1}$ might be vertex $b$, and $v_{i}$ might concide with $u_{i+1}$ for some $i$.

We have that both $C'(u_{i},v_{i})$ and $\pi(u_{i},v_{i})$ are shortest paths and then they have the same length, for $i=1, \ldots, \ell -1$. Therefore, the cycle $C_{1}$ obtained from $C$ by substituting $C'(u_{i},v_{i})$ with $\pi(u_{i},v_{i})$ for any $i$ odd, with $1\leq i \le \ell - 2$, is a minimum $st$-separating cycle. Similarly,  the cycle $C_{2}$ obtained from $C$ by substituting $C'(u_{i},v_{i})$ with $\pi(u_{i},v_{i})$ for any $i$ even, with $2\leq i \le \ell - 1$, is a minimum $st$-separating cycle. A general example is shown in Figure~\ref{fi:crossingsplit}.

Both $C_{1}$ and $C_{2}$ are single-crossing minimum $st$-separating cycles and each edge in $C$ is contained in at least one among $C_{1}$ and $C_{2}$. 
\end{proof}

\begin{figure}[t]
\begin{center}
\def\svgwidth{12cm}
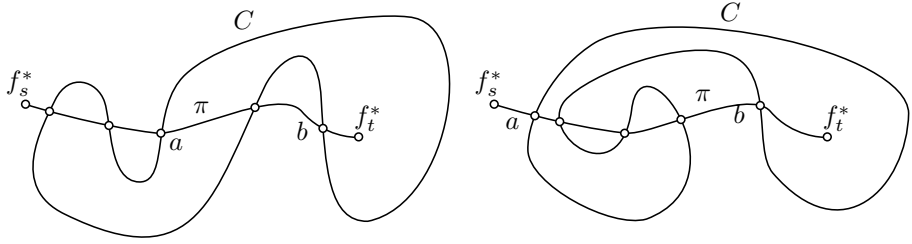
\end{center}
\caption{Forbidden crossings: in both cases, $C$ cannot be a minimum $st$-separating cycle.}\protect\label{fi:crossingsforbidden}
\end{figure}	

\begin{figure}[t]
\begin{center}
\def\svgwidth{12cm}
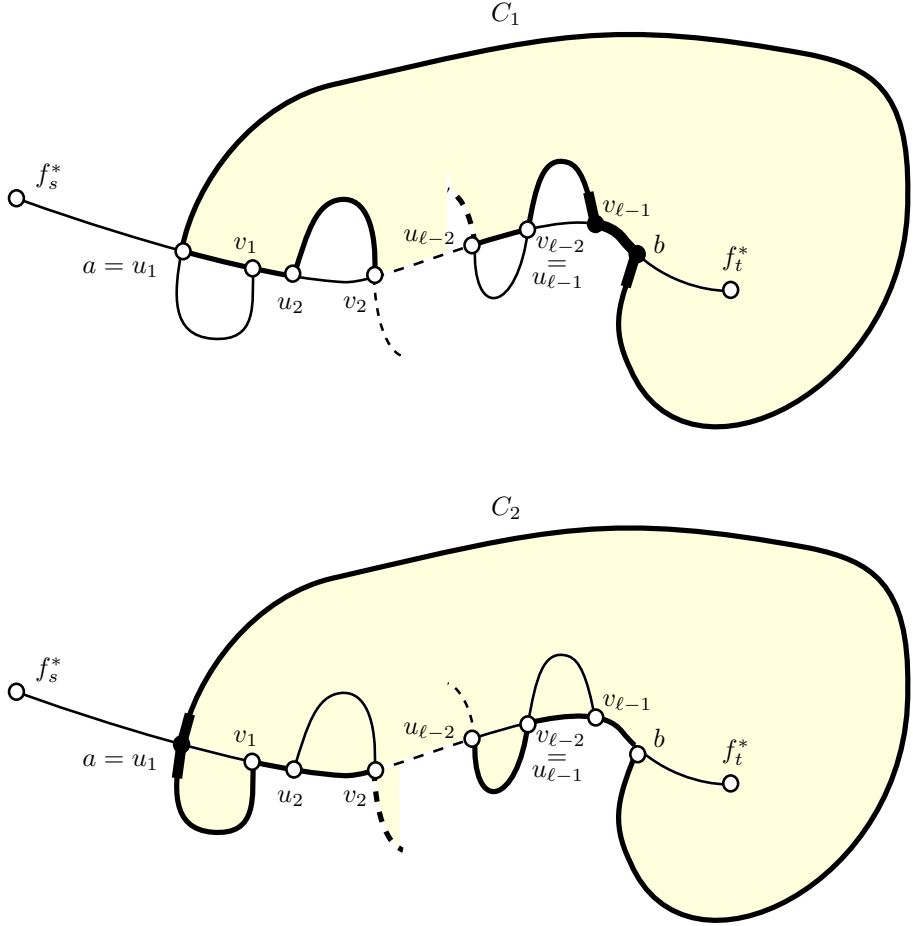
\end{center}
\caption{Two new minimum $st$-separating cycles: $C_{1}$ only crosses $\pi$ around path $v_{\ell}-1$, $b$, while $C_{2}$ only crosses $\pi$ in $a$. Crossings are denoted by black dots and paths. Each edge in $C$ is in $C_{1}$ or $C_{2}$.}\protect\label{fi:crossingsplit}
\end{figure}

In order to find edges that belong to single-crossing minimum $st$-separating cycles, we propose a two-phase approach.
In the first phase, we fix a shortest $s^{*}t^{*}$-path $\pi$ in the dual graph, and find all vertices on $\pi$ that belong to at least one minimum $st$-separating cycle (call $N$ the set of such vertices). In the second phase, beginning with the median vertex $u$ in $N$, with respect to the order in which vertices appear in $\pi$, we find all edges that belong to at least one minimum $st$-separating cycle through $u$, and then proceed recursively, taking as input instances the portion of the dual graph inside the innermost/outside the outermost minimum $st$-separating cycle through $u$, respectively.

\subsection{Overview of Reif's algorithm}

First of all, let us briefly recall Reif's algorithm (see Figure~\ref{fi:reif}). According to his approach, graph $G^{*}$ is ``cut'' along a shortest path $\pi$ from $f_s^{*}$ to $f_t^{*}$, obtaining graph $G^{c}$, in which each vertex $f^{*}$ in $\pi$ is split into two vertices $f^{*}_{u}$ and $f^{*}_{\ell}$. Edges in $G^{*}$ incident on each $f^{*}$  from above $\pi$ are moved to $f^{*}_{u}$ and edges incident on $f^{*}$ from below $\pi$ are moved to $f^{*}_{\ell}$. Edges incident on $f^{*}_{s}$ and $f^{*}_{t}$ are considered above or below $\pi$ on the basis of two dummy edges joining $f^{*}_{s}$ to a dummy vertex inside face $s^{*}$  and  $f^{*}_{t}$ to a dummy vertex inside face $t^{*}$.
Edges on $\pi$ are copied both on $f^{*}_{u}$ and on $f^{*}_{\ell}$, so that path $\pi$ is doubled. Property~\ref{pro:crossundir} shows that, for each vertex $f^*$ in $\pi$, a shortest cycle $\gamma_{f}$ among all the $st$-separating cycles containing $f^{*}$ corresponds to a shortest path in $G^c$ from $f^*_u$ to $f^*_\ell$.
The algorithm finds the median vertex $f^{*}$ in $\pi$ and computes $\gamma_{f}$ by a SSSP algorithm starting from $f^{*}_{u}$, finding a shortest path in $G^c$ from $f^*_u$ to $f^*_\ell$. Cycle $\gamma_{f}$ divides $G^{c}$ in a \emph{left} part $G_{\ell}^{c}$ and a \emph{right} part $G_{r}^{c}$ (assuming $s$ is drawn to the left of $t$), where the boundary $\gamma_{f}$ belongs to both $G_{\ell}^{c}$ and $G_{r}^{c}$. The same technique is applied recursively to $G_{\ell}^{c}$ and $G_{r}^{c}$. This gives rise to a recursion tree with depth $O(\log n)$, where the sum of the sizes of all instances on a same level is $O(n)$, provided that paths containing vertices of degree two 
common to several instances are represented in a compact way, as shown in~\cite{Reif}.  Thus, the overall worst-case time to find the shortest $st$-separating cycle in $G^{*}$ is $O(\mbox{SSSP}(n) \cdot  \log n)$, where $\mbox{SSSP}(n)$ is the time needed to compute a SSSP tree in a planar graph. The complexity of Reif's algorithm becomes $O(n \log n)$ if shortest path trees are computed in $O(n)$ time as in~\cite{Henzinger}.

\begin{figure}[t]
\begin{center}
\def\svgwidth{8cm}
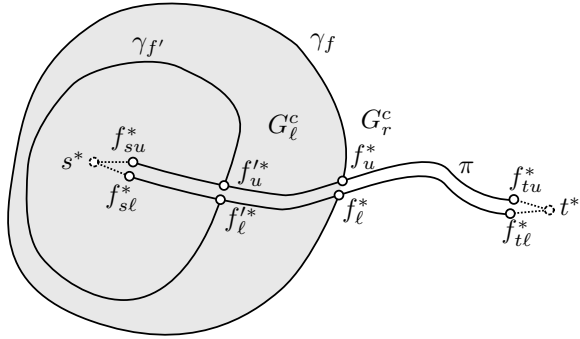
\end{center}
\caption{Reif's divide and conquer technique. Graph $G^{c}$ is split by $\gamma_{f}$ into $G^{c}_{\ell}$ and $G^{c}_{r}$. A shortest $st$-separating cycle $\gamma_{f'}$ can be found in $G^{c}_{\ell}$.}\protect\label{fi:reif}
\end{figure}	

In Section~\ref{subsec:desc_of_algo}, we describe how the divide and conquer algorithm by Reif~\cite{Reif} can be adapted to compute the vitality of all edges in an undirected unweighted planar graph.

\subsection{Description of our algorithm}
\label{subsec:desc_of_algo}
As stated above, our algorithm works in two phases.

In the first phase, we apply Reif's approach.
We choose a shortest $s^{*}t^{*}$-path $\pi$ in the dual graph, and we note that for each vertex $f^{*}$ in $\pi$, Reif's algorithm computes the length of a shortest $st$-separating cycle containing $f^{*}$. Hence it is possible to determine the set of vertices in $\pi$ which belong to at least one minimum $st$-separating cycle.
After duplicating path $\pi$, each vertex $f_{i}^{*}$ is replaced by a pair of vertices $(x_i, y_i)$, and a minimum $st$-separating cycle corresponds to a path joining $x_i$ and $y_i$. Reif's algorithm allows us to compute $\mbox{dist}(x_i, y_i)$, for each $i$. Let $K$ be the set of indices $i$ for which $\mbox{dist}(x_i, y_i)$ is minimum; set $K$ identifies all the vertices in $\pi$ which belong to at least one shortest $st$-separating cycle.

In the second phase, we choose a planar embedding of the dual graph such that $x_i$ and $y_i$ are on the outer face, $\forall i\in K$. For each edge $e^{*}$ in the dual graph, we determine whether it belongs to a shortest path between at least one pair $(x_i, y_i), i\in K$. If this is the case, then $\mbox{vit}(e)=1$, otherwise $\mbox{vit}(e)=0$.
In order to better explain this phase, let us first make the simplifying assumption that there is only one pair $(x_i, y_i)$ at minimum distance (i.e., $K=\{i\}$). We are now going to define the concept of \emph{leftmost} (\emph{rightmost}) BFS visit~\cite{KleinSoda2005, WagnerJComp1997}, which we will use extensively in the following.

\paragraph{Leftmost (rightmost) BFS visit of a planar graph.} Given a planar embedding of the input graph, a leftmost BFS visit will process unvisited neighboring vertices in counterclockwise order, starting from the vertex immediately following (in counterclockwise order) the one from which the current vertex was reached. Analogously, a rightmost BFS visit will process unvisited vertices in clockwise order, starting from the vertex immediately following (in clockwise order) the one from which the current vertex was reached. Let the shortest path between vertices $x_i$ and $y_i$ in a leftmost BFS visit starting from $x_i$ be the \emph{leftmost} $(x_i, y_i)$-shortest path (designated as $\pi^\ell_i$). We analogously define $\pi^r_i$ as the \emph{rightmost} $(x_i, y_i)$-shortest path starting from $x_i$.

\medskip
\begin{figure}[t]
\begin{center}
\def\svgwidth{8.5cm}
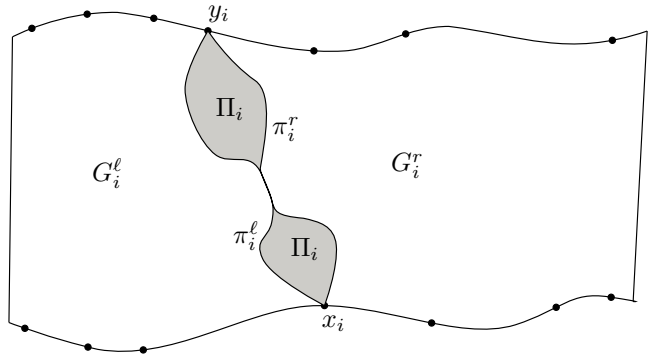
\end{center}
\caption{The leftmost and rightmost shortest paths $\pi_i^{\ell}$ and $\pi_i^r$ define the three subgraphs $G_i^{\ell}$, $\Pi_i$ and $G_i^r$. Paths $\pi_i^\ell$ and $\pi_i^r$     may share vertices and/or edges.}\protect\label{fig:losanga}
\end{figure}	

We designate the subgraph induced by vertices contained in the region between $\pi^\ell_i$ and $\pi^r_i$ (included) as $\Pi_i$. An example is shown in Figure~\ref{fig:losanga}.
Since any $(x_i,y_i)$-shortest path lies to the right of $\pi^\ell_i$ and to the left of  $\pi^r_i$, all $(x_i,y_i)$-shortest path are contained in $\Pi_i$.
We now highlight the following:

\medskip

\begin{remark}
\label{rem:bfs}
For each edge $(u,v)\in \Pi_i$, in order to determine whether it belongs to a shortest path between $x_i$ and $y_i$, it suffices to check if either $\mbox{dist}(u, x_i)+\mbox{dist}(v, y_i)+1=\mbox{dist}(x_i, y_i)$ or $\mbox{dist}(u, y_i)+\mbox{dist}(v, x_i)+1=\mbox{dist}(x_i, y_i)$.
\end{remark}

\medskip

The above remark suggests how $\Pi_i$ should be processed. We run two BFS visits, from $x_i$ and from $y_i$ respectively, and then, for each edge $e \in \Pi_i$, we check if one of the above conditions holds. If this is the case then $\mbox{vit}(e)=1$, otherwise $\mbox{vit}(e)=0$.

Let us now lift the above simplifying assumption, and assume that $|N|>1$. In this case, let $(x_i, y_i)$ be the vertices corresponding to the median vertex in $N$ along $\pi$, and proceed as above. The leftmost and the rightmost shortest paths between $x_i$ and $y_i$ divide the graph into three regions: the subgraph induced by vertices left of $\pi^\ell_i$, including $\pi^\ell_i$ (call it $G^\ell_i$), the subgraph induced by vertices right of $\pi^r_i$, including $\pi^r_i$ (call it $G^r_i$), and the subgraph $\Pi_i$ contained between $\pi^\ell_i$ and $\pi^r_i$, including $\pi^\ell_i$ and $\pi^r_i$ themselves. See Figure~\ref{fig:losanga} for an example.

$\Pi_i$ can be processed according to Remark~\ref{rem:bfs}. The algorithm then proceeds recursively, with $G^\ell_i$ and $G^r_i$ as input instances. We will now show the correctness of this approach.
More in detail, we show the following result:
\begin{lemma}
\label{lemma-decomposition}
Let $i$ be any index in $K$:
\begin{enumerate}[(i)]
\item\label{statement_1} shortest paths between $x_j$ and $y_j$, with $j>i$, can only be in $G^r_i \cup \Pi_i$;
\item\label{statement_2}  for each pair of vertices in $G^r_i$, their distance in $G^r_i \cup \Pi_i$ is the same as their distance in $G^r_i$;
\item\label{statement_3} portions of $(x_j, y_j)$-shortest paths in $\Pi_i$ are also in some $(x_i, y_i)$-shortest paths.
\end{enumerate} 
\end{lemma}

\begin{proof}
Let $(x_j, y_j)$ be any pair in instance $G^r_i$. We first observe that all $(x_j, y_j)$--shortest paths must be contained in $G^r_i \cup \Pi_i$, otherwise $\pi^\ell_i$ would not be the leftmost $(x_i, y_i)$--shortest path. This proves~(\ref{statement_1}).

Let us now assume that a portion of path $\pi_j$ is contained in $\Pi_i$ (see Figure~\ref{fig:decomposition}). Let $u$ be the first vertex of $\pi_j$ in $\Pi_i$, and let $v$ be the last. Clearly, the $\pi^{uv}_j$ subpath of $\pi_j$ can be safely replaced by the $uv$ subpath of $\pi^r_i$ without affecting distances for all vertices of $\pi_j$ in $G^r_i$.
This proves~(\ref{statement_2}).

\begin{figure}[t]
\begin{center}
\def\svgwidth{12.5cm}
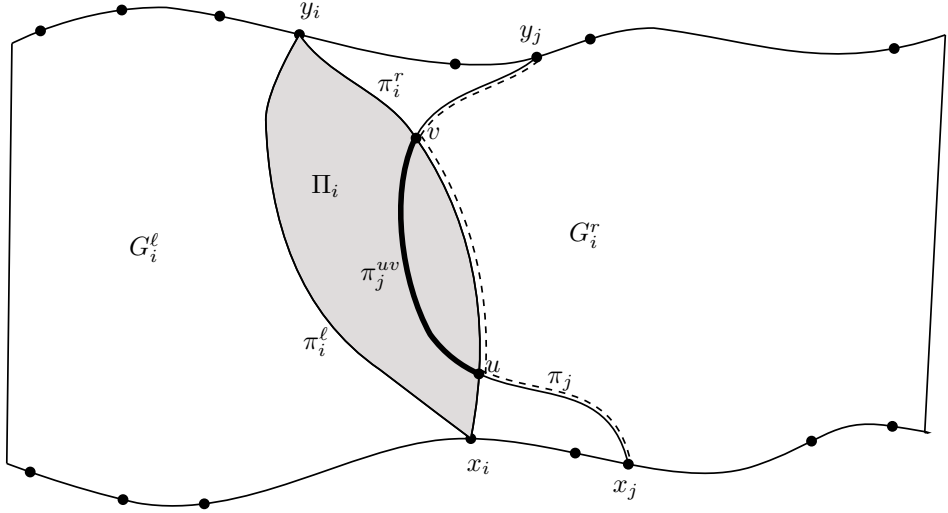
\end{center}
\caption{If a shortest path from $x_j$ to $y_j$, with $j > i$, crosses $\pi_i^r$ then another shortest path from $x_j$ to $y_j$ exists in $G_{i}^{r}$.}\protect\label{fig:decomposition}
\end{figure}	

Finally, we observe that $\pi^{uv}_j$, the subpath we replaced, being a shortest path between two vertices in $\Pi_i$, must have already been examined, with respect to the vitality of its component edges, in a previous step of the recursion. This proves~(\ref{statement_3}).
\end{proof}

Besides the subgraph it will process, each recursive call receives as input the subset of pairs of vertices on the outer face between which to compute shortest paths. We take as the base case of the recursion any instance whose input only contains one or two such pairs, analogously to what is done in~\cite[p. 615]{Hassin85}. In this case, we apply the process described in Remark~\ref{rem:bfs} to each pair.

\subsection{Time and space bounds} We now proceed to bound time and space requirements for our algorithm.

The first phase closesly parallels Reif's algorithm~\cite{Reif}, whose running time, 
as pointed out in the introduction, can be reduced to $O(n \log n)$ by applying the SSSP tree algorithm by Henzinger \emph{et al.}~\cite{Henzinger}. Differently from Reif, however, we are solving the problem for unweighted graphs, which leads to an $O(n \log n)$ running time for the first phase even without resorting to the above cited sophisticated SSSP tree algorithm.

In order to bound the running time of the second phase, we must compute the total size of all subinstances at each level of recursion. Every subinstance $\chi$ (let us call subinstance $\chi$ the instance whose median pair is $(x_\chi, y_\chi)$) is delimited on the left by a rightmost shortest path, say $\pi^r_i$, and on the right by a leftmost shortest path, say $\pi^\ell_j$, with $i<\chi<j$ (we assume that there are two dummy paths of minimum length to the left and to the right of the extremal subinstances). The following analysis runs along the same lines as Theorem 5 in~\cite{Reif}. Because of the way the graph is split into $\Pi_i$, $G^\ell_i$ and $G^r_i$, each edge belongs to at most two instances on the same level of recursion, unless it is part of a ``degenerate'' frontier: degenerate frontiers occur when $\pi^r_i$ and $\pi^\ell_j$ share some edges, these edges must all be consecutive, because frontier paths are leftmost and rightmost, respectively.
This degenerate frontier, consisting of a path of vertices all of degree 2 in the current instance, will be represented as a single appropriately weighted edge (see Figure~\ref{fig:degenere}).

\begin{figure}[t]
\begin{center}
\def\svgwidth{8.5cm}
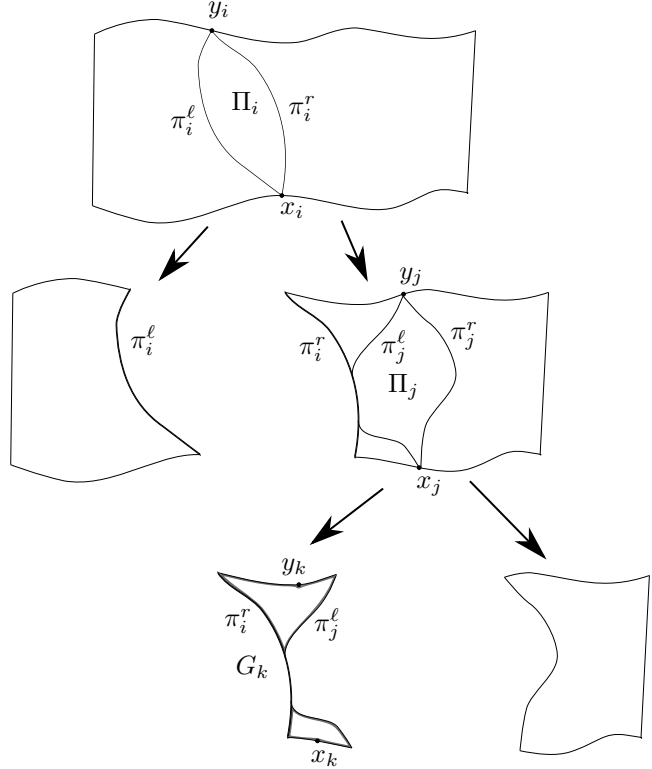
\end{center}
\caption{A degenerate instance $G_{k}$, bounded by a rightmost shortest path $\pi_i^r$ and a leftmost shortest path $\pi_j^\ell$.}\protect\label{fig:degenere}
\end{figure}	

At level $k$ of the recursion, therefore, a total of $O(m)+O(2^k)$ edges will be processed by our algorithm. The first term accounts for unduplicated edges, plus simple frontier edges that are duplicated only once. Since the graph is planar, this term can be reduced to $O(n)$. The second term accounts for the weighted edges resulting from degenerate frontiers between instances, since each instance can contain at most one such edge and there are at most $2^k$ instances at level $k$. Since the levels of recursion are $k \leq \log_{2} n$, it follows that $2^k \leq n$. Therefore the total number of processed edges at each level is $O(n)$.

The processing performed by each recursive call consists of two BFS visits (leftmost and rightmost), which take time linear in the number of edges, even when taking into account the degenerate frontier weighted edges (at most one per instance), plus the verification of whether each edge is on a shortest path or not, as described in Remark~\ref{rem:bfs}, which requires constant time per edge. For base case instances, we only need to perform a BFS visit starting from each vertex on the outer face. Therefore the second phase of our algorithm runs in $O(n)$ time per level of recursion, and thus in $O(n\log n)$ total time.

If each recursive call were provided a copy of the relevant portion of the original graph as input, the space requirement of our algorithm would be $O(n\log n)$. However, this bound can be reduced to $O(n)$ by observing that all recursive calls can run on the original graph, provided that, before calling the procedure on left and right subinstances, edges to be ignored are properly marked, and they are unmarked when the call returns.

The above discussion leads to the following theorem:

\begin{theorem}
Given an unweighted undirected planar graph with $n$ vertices, the vitality of all edges with respect to the max-flow between two fixed vertices $s$ and $t$, can be computed in $O(n \log n)$ worst-case time and $O(n)$ worst-case space.
\end{theorem}


\section{Conclusions and further work}
\label{se:conclusions}

In this paper we have shown how to efficiently compute the vitality of all edges with respect to max-flow, in the case of undirected unweighted planar graphs. Note that our result also holds for uniform edge weights. In a previous work~\cite{AusielloNetworks}, the same authors presented algorithms for solving the vitality problem for all edges, in the case of general undirected graphs, and for all arcs and nodes, plus contiguous arc sets, for $st$-planar graphs, both directed and undirected.

Several points remain open, most notably the case of planar weighted undirected graphs, for which it is not sufficient to look for minimum $st$-separating cycles in the dual graph, hence it is not possible to rely on the divide and conquer technique by Reif. The technique described in this paper cannot be applied to planar directed graphs either, even in the unweighted case, due to the fact that it is not always possible to avoid crossings between shortest $st$-separating cycles.

\section*{Acknowledgements}
The authors are grateful to an anonymous referee for their valuable suggestions.

\bibliographystyle{plain}
\bibliography{vital}

\end{document}